\documentclass[10pt, conference, a4paper]{IEEEtran}

\usepackage[T1]{fontenc}% optional T1 font encoding
\usepackage{cite}
\usepackage{url}
\usepackage{amsmath}
\usepackage{amssymb}
\usepackage{amsthm}
\usepackage{empheq}
\usepackage[cmintegrals]{newtxmath}
\usepackage{tabularx}
\usepackage{multirow}
\usepackage{booktabs}
\usepackage{diagbox}

\usepackage[labelformat=simple]{subcaption}

\usepackage{dblfloatfix}
\usepackage{graphicx}
\usepackage{epstopdf}
\usepackage{xcolor}
\usepackage{mathtools,amssymb,lipsum}
\usepackage[detect-weight=true, binary-units=true]{siunitx}
\usepackage[inline]{enumitem}
\usepackage{cuted}
\usepackage{blindtext}
\usepackage[boxruled,norelsize]{algorithm2e}
\usepackage{float}
\usepackage{siunitx}

\newtheorem{theorem}{Theorem}

\newcommand\blfootnote[1]{%
	\begingroup
	\renewcommand\thefootnote{}\footnote{#1}%
	\addtocounter{footnote}{-1}%
	\endgroup
}

\newcommand{\change}[1]{#1}
\renewcommand{\change}[1]{\textcolor{black}{#1}}

\newcommand{\transpose}[0]{^\text{T}}
\makeatletter
\newcommand{\removelatexerror} {\let\@latex@error\@gobble}
\makeatother

\begin{document}

\title{A Utility-Driven Multi-Queue Admission Control Solution for Network Slicing}
\author{
	\IEEEauthorblockN{Bin Han\IEEEauthorrefmark{1}, Vincenzo~Sciancalepore\IEEEauthorrefmark{2}, Di~Feng\IEEEauthorrefmark{3}, Xavier~Costa-Perez\IEEEauthorrefmark{2} and Hans D. Schotten\IEEEauthorrefmark{1}\IEEEauthorrefmark{4}}\\
	\IEEEauthorblockA{
		\begin{minipage}[t]{.44\textwidth}
			\centering
			\IEEEauthorrefmark{1}Technische Universit\"at Kaiserslautern, Germany\end{minipage}
		\hfill
		\begin{minipage}[t]{.53\textwidth}
			\centering
			%		\IEEEauthorblockA{
			\IEEEauthorrefmark{2}NEC Laboratories Europe, Germany\end{minipage}\vspace{1mm}\\
		\begin{minipage}[t]{.44\textwidth}
			\centering
			%		\IEEEauthorblockA{
			\IEEEauthorrefmark{3}Universitat Aut\`onoma de Barcelona, Spain\end{minipage}
		\hfill
		\begin{minipage}[t]{.53\textwidth}
			\centering
			%		\IEEEauthorblockA{
			\IEEEauthorrefmark{4}DFKI GmbH, Germany\end{minipage}	}
}

%For double-blind review
%\author{
%	\IEEEauthorblockN{Author 1\IEEEauthorrefmark{1}, Author 2\IEEEauthorrefmark{2}, Author 3\IEEEauthorrefmark{3}, Author 4\IEEEauthorrefmark{2} and Author 5\IEEEauthorrefmark{1}\IEEEauthorrefmark{4}}\\
%	\IEEEauthorblockA{
%	\begin{minipage}[t]{.45\textwidth}
%		\centering
%		\IEEEauthorrefmark{1}Institution 1\\
%		Institution Address 1\\
%		\{author1, author5\}@institution1.xxx
%	\end{minipage}
%	\hfill
%	\begin{minipage}[t]{.45\textwidth}
%		\centering
%		\IEEEauthorblockA{
%		\IEEEauthorrefmark{2}Institution 2\\
%		Institution Address 2\\
%		\{author2,author4\}@institution2.xxx
%	\end{minipage}\vspace{3mm}\\
%	\begin{minipage}[t]{.45\textwidth}
%		\centering
		%		\IEEEauthorblockA{
%		\IEEEauthorrefmark{3}Institution 3\\
%		Institution Address 3\\
%		author3@institution3.xxx
%	\end{minipage}
%	\hfill
%	\begin{minipage}[t]{.45\textwidth}
%		\centering
%		\IEEEauthorrefmark{4}Institution 4\\
%		Institution Address 4\\
%		author4@institution4.xxx
%	\end{minipage}	}
%}




\maketitle

\begin{abstract}

%XC The emerging techniques of network function virtualization (NFV) and network programmability (SDN) give rise to the novel network slicing paradigm that enables the Fifth Generation (5G) mobile communication networks to run a multi-tenant business over rented ``slices'' of network, known as \emph{Slice as a Service} (SlaaS). Existing studies have demonstrated the flexibility of this business model and its potential in network utility optimization, without any concern of end-service delay or slice admission rate.

%XC
The combination of recent emerging technologies such as network function virtualization (NFV) and network programmability (SDN) gave birth to the Network Slicing revolution. 5G networks consist of multi-tenant infrastructures capable of offering leased network "slices" to new customers (e.g., vertical industries) enabling a new telecom business model: Slice-as-a-Service (SlaaS). In this paper, we aim $i$) to study the slicing admission control problem by means of a multi-queuing system for heterogeneous tenant requests, $ii$) to derive its statistical behavior model, and $iii$) to provide a utility-based admission control optimization. Our results analyze the capability of the proposed SlaaS system to be approximately Markovian and evaluate its performance as compared to legacy solutions.
%XC
\end{abstract}


% \fi
\IEEEpeerreviewmaketitle

\section{Introduction}

%XC The huge interest around the softwarization and virtualization of network deployments is attracting different vertical segments---such as Over-The-Top (OTT) applications---alien to this innovative research challenge in the recent years. Such technological means give rise to the a novel paradigm: the \emph{network slicing}~\cite{intro_slic}. Envisioned as an emerging technology initiative with the ability of letting an infrastructure provider to offer a ``slice'' of the network resources (in terms of computational, storage and networking) to network tenants, it has been deeply investigated in the last few years as the main driver towards the next generation of mobile networks, dubbed as 5G networks.

\emph{Network Slicing}~\cite{intro_slic}
is an emerging 5G technology that allows infrastructure providers to offer ``slices'' of resources (computational, storage and networking) to network tenants. In this way a new business game~\cite{networkslic}
is introduced as infrastructure providers (sellers) strategically decide which tenants (buyers) get granted slices to deliver their services. Intuitively, this involves a number of challenges that fall in the economic research field, which, in turn, requires a detailed understanding of the context. In particular, the infrastructure provider may rely on this emerging technology as a means to increase its revenue sources. However, to achieve the overall revenue maximization, advanced admission control policies are required as tenants compete for a limited bunch of available resources.
\blfootnote{This is a preprint. \copyright~2019 IEEE}
%XC Network slicing builds up a new business game~\cite{networkslic} where the infrastructure provider (the seller) might decide which tenants (the buyers) might take over the control of the slice to deliver their own services. Intuitively, this involves a number of challenges that fall in the economic research field, which, in turn, requires a detailed understanding of the context. In particular, the infrastructure provider may rely on this emerging technology as a means to bring additional sources of revenue but, given the scarcity of access resources, it might require advanced admission control policies to pursue the overall revenue maximization.

In this competing environment, a brokering solution may act as a mediator between seller and buyers while providing service level agreements (SLAs) guarantees to granted running slices~\cite{samdanis2016network}. %This has been mathematically modeled as a Markov decision process (MDP) in~\cite{bega2017optimising}, where Q-Learning helps to reduce the complexity of the Markov Decision Process while still providing near-optimal results in affordable time.
Admission control policies will guide the broker in the process of deciding the set of network slices that can be installed on the system and the ones to be rejected. As the number of network slices grows---as envisioned for the next few years~\cite{slicMobicom2018}---it will be necessary to design an automated solution that dynamically decides on the received slice requests while guaranteeing a certain degree of fairness among network tenants. 
Indeed, network slice requests may be queued while waiting for the next available resources, or may be re-issued. 

To properly design such a \emph{slicing brokering process}, a deep understanding of the slice queuing behavior is needed that accounts, for e.g. the average slice duration (based on the slice type), the frequency of slice requests (based on the tenant), etc. This enables a Slice-as-a-Service (SlaaS)~\cite{sciancalepore2017slice} solution that fully supports on-demand slices requests: tenants issue slice requests for given periods of time and decide whether to re-issue the same request upon rejection based on service level agreements. Advanced slicing admission control solutions may have different policies for tenants frequently asking for short-term slices---such as Internet-of-Things (IoT), or crowded event-based network slices---as they will automatically re-issue the same request in the near future, with respect to those that require only few longer network slices---such as Mobile Virtual Network Operators (MVNOs) or Industrial Network Slices~\cite{TII2018_slicing}---which may be probably lost if not accepted. %\,\footnote{Economic loss might be modeled by means of advanced slicing scheme. However, this is out of the scope of this paper.}. 
Moreover, similar as widely recognized in all kinds of queuing systems for service scheduling, tenants may be \emph{impatient} and choose to leave for another available infrastructure provider instead of waiting in queue, especially when the expected waiting time is long. Such behavior shall also be taken into account while designing a slicing admission control solution to mitigate potential revenues loss in case of resource congestion.

While conventional admission control problems have been extensively studied in the literature, we pioneer a new stochastic model for network slicing that leverages on the multi-queuing system to optimally design an admission control of on-demand network slices as well as to orchestrate them once are accepted. This also allows to account for impatient tenant behaviors and heterogeneous network slice characteristics while, at the same time, enforcing given performance metrics, such as fairness between different tenants or between network slice types or utility-based maximization.

%The remainder of this paper is structured as follows. In Section~\ref{sect:model}, we introduce our assumptions and we formulate the network slicing admission control problem for on-demand slice request arrivals. In Section~\ref{sect:ns_queue}, we provide a simple use case to cast our problem into a queuing system. In Section~\ref{sect:multi-queue-controller}, we model the problem as a multi-queue problem where each queue may host slice requests of the same type while waiting for being granted. In Section~\ref{sect:controller},
%we devise the multi-queue controller accounting for additional metrics,
%we devise and analyze the multi-queuing controller with additional metrics by proving the capability of conventional queuing models on such system,
%whereas in Section~\ref{sect:optimization} we propose our optimization strategy. In Section~\ref{sect:perf_eval}, we carry out an exhaustive simulation campaign to prove our findings and validate our model. In Section~\ref{sect:discussions} we discuss the applicability of some assumptions whereas in Section~\ref{sect:rel_work} we outline the main related works on this topic. Finally, in Section~\ref{sect:concl} we provide concluding remarks.

\section{Model design}
\label{sect:model}
We cast our problem into a typical network slicing scenario, where the Mobile Network Operator (MNO) decides to lease infrastructure resources to tenants, willing to pay to take over the control of an independent network slice so as to deliver an end-service to their own users.
Hereafter, we deeply describe our assumptions and mathematically formulate the problem.
\subsection{Resource pool and slice types}
Let us consider a single MNO that possesses a static resource pool of $M$ different resources and offers $N=|\mathcal{N}|$ pre-defined types of slices. Depending on the slice type $n\in\mathcal{N}$, it costs a certain resource bundle to create and maintain a slice.  Let $\mathbf{r}=[r_1,r_2,\dots,r_M]^\text{T}$, $\mathbf{s}=[s_1,s_2,\dots,s_N]^\text{T}$ and $\mathbf{c}_n=[c_{1,n},c_{2,n},\dots,c_{M,n}]^\text{T}$ denote the resource pool, the set of slices under maintenance and the resource bundle required to maintain a slice of type $n\in\mathcal{N}$, respectively. The assigned resources can be then represented as
\begin{equation}
	\mathbf{a}\overset{\Delta}{=}[a_1,a_2,\dots,a_M]^\text{T}=\mathbf{C}\times\mathbf{s},
\end{equation}
where $\mathbf{C}=[\mathbf{c}_1,\mathbf{c}_2,\dots,\mathbf{c}_N]$. 
At any time instance, the MNO cannot simultaneously maintain more slices than its resource pool may support. This constraint is expressed using the \emph{space of resource feasibility}\cite{han2018slice}:
\begin{equation}
	\mathbb{S}=\{\mathbf{s}\vert r_m-a_m\ge 0,\quad\forall 1\le m\le M\}. 
\end{equation}
Note that $\mathbb{S}$ is a finite discrete set, thus the MNO can be characterized as a finite state machine where each slice set under maintenance represents the system state $\mathbf{s}\in\mathbb{S}$.

\subsection{Slice admission in SlaaS}
We consider a certain number of tenants randomly generating network slice requests. Slices requested by a certain tenant are of the same type. For each tenant, the inter-arrival time between two requests is drawn from an exponential distribution. The request arrivals of different tenants are independent and identically distributed (i.i.d.).

Once a request for slice creation is triggered, the MNO makes a binary decision, i.e., the MNO either accepts or declines it. Upon acceptance, the requested slice is created, and continuously maintained so that a corresponding bundle of network resources is occupied until the slice is terminated (at the end of its lifetime) and the resource bundle is released. It should be noted that the constraint of space of resource feasibility forbids the MNO to accept any request when its current state is close to the border of $\mathbb{S}$. In other words, if the current MNO resource pool is close to be saturated by active slices, it does not accept additional network slice requests that might experience a service disruption.
This introduces the well-known concept of \emph{admissibility region}\footnote{The admissibility region has been exhaustively studied in the literature for different use cases and scenarios. We refer the reader to~\cite{bega2017optimising}, where a stochastic admissibility region is derived for a network slicing admission control.} described as
\begin{equation}
	\mathbb{A}=\{\mathbf{s}\vert\mathbf{s}\in\mathbb{S},\exists n:\mathbf{s}+\Delta\mathbf{s}_n\in\mathbb{S}\},
\end{equation}
where $\Delta\mathbf{s}_n$ is the \emph{unit slice incremental vector} of type $n$
\begin{equation}\label{equ:slice_incremental_vector}
\Delta\mathbf{s}_n=[\underbrace{0,\dots,0}_{n-1},1,\underbrace{0,\dots,0}_{N-n}],\quad n\in\{1,2,\dots,N\}.
\end{equation}
We assume that the lifetime of every slice is an i.i.d. exponentially distributed variable and the expected lifetime depends on the slice type. We also consider that the MNO makes every decision according to a consistent slicing policy, i.e., the decision depends only on the type of requested slice $n$ and the current system state $\mathbf{s}$ that defines the current set of slices under maintenance.

\subsection{Delayed reattempt upon request denial}
If a request for slice creation is declined---because of a temporary shortage of available resources due to many other active slices---the tenant is not able to obtain the requested slice immediately. Instead, its request may be sent to the MNO again for a reconsideration after some delay with the hope that some running slice has expired (i.e., resources have been released). Generally, there are two critical features of the delaying mechanism, which should be taken into account: $i$) resource efficiency and $ii$) fairness. The former requires that the chosen mechanism purses the resource pool utilization maximization whereas the latter requires that the expected delay for different requests is normalized.

Two categories of approaches are commonly used to solve this kind of problem:

\noindent\textbf{Random delay}. Every declined request is re-proposed to the MNO after a random delay. This approach provides a good fairness, but generates extra signaling overhead in the control plane being not able to provide the discipline of ``First Come, First Served'' (FCFS), as described in the next section.

\noindent\textbf{Queuing}. Declined requests wait in one or multiple queue(s) for the next opportunity during the MNO's decisional process. This is the most common solution in cloud service scheduling.

Hereafter, we show how a multi-queuing system may be fully exploited to provide insights on the system behaviors and pave the road towards a slicing orchestration solution.

%\section{Slice Set as a Markov Process}
%According to the previous work (submitted to GLOBECOM'19), if resource pool and slice types are static, the request arrivals and slice releases are \textcolor{red}{\textbf{consistent}} memoryless random processes, and the slice strategy is consistent, the set of slices under maintenance is a Markov process.
%
%\section{The Random Access Method}
%\begin{itemize}
%	\item If the accepting rate exceeds the arriving rate (incl. reattempt), it converges to a dynamic balance.
%	\item If the accepting rate remains lower than the arriving rate (incl. reattempt), declined requests will stack to a infinite number, making the arrivals inconsistent - so a maximal number of reattempts is needed.
%	\item One can obtain (calculate?/simulate) the average waiting time and acceptance rate as a function of the arriving rate (excl. reattempts), the accepting rate, the random delay distribution and the maximal reattempts.
%\end{itemize}

\section{Network slicing queuing}
\label{sect:ns_queue}
In the literature a number of various disciplines have been studied to serve the request queues. Among the others, the most common policies are $i$) First come, first served (FCFS), $ii$) Last come, first served (LCFS), $iii$) Random selection for service (RSS) and $iv$) Priority-based (PR). All of them analyze different behaviors and are used to achieve distinct performance metrics. For instance, the LCFS is used to reduce the fairness whereas the priority-based is implemented when there is some high-level preference of the MNO to be considered. RSS shows huge complexity in the implementation without bringing any significant advantage with respect to the others. Hereafter, we focus on the FCFS case. However, any other discipline may be easily adapted to our analysis.
\subsection{Queuing schemes}
We differentiate the queuing systems into two different categories: $i$) single-queue and $ii$) multi-queue systems. When considering the single-queue, only one queue is implemented for all declined requests that need to wait for the next acceptance opportunity. Conversely, the multi-queue system implements multiple queue for declined requests. Specifically, such queues may show different features. We consider homogeneous-mixed queues, wherein each queue consists of requests for slices of different types, and heterogeneous queues, where each queue is specified for only one unique slice type. We next show a simple case-study to justify that the queuing system is suitable for this kind of problems.

\subsection{Resource efficiency: a simple case-study}
Consider a simplified case where $M=1$, $N=2$, $\mathbf{r}=[1]$, $\mathbf{c}_1=[0.6]$, $\mathbf{c}_2=[0.2]$ and $\mathbf{s}=[1,0]\transpose$. The first four requests awaiting in the queue(s) are in the sequential order $[1,1,2,2]$. The MNO is taking a greedy strategy that intends to accept all requests received so far the resource pool supports.

Both in the schemes with a single queue and two homogeneous queues, the MNO fails to accept requests of type~$2$ as the type~$1$ requests are preventing their acceptance. Hence, it has to wait until the currently active slice of type~$1$ is released before it can accept the next request in the queue, although it has both enough idle resource and the intention. The heterogeneous multi-queue scheme, in contrast, enables the MNO to fully utilize its resource pool as shown in Fig.~\ref{fig:schemes}.
\begin{figure}[!htbp]
	\centering
	\includegraphics[width=.5\textwidth]{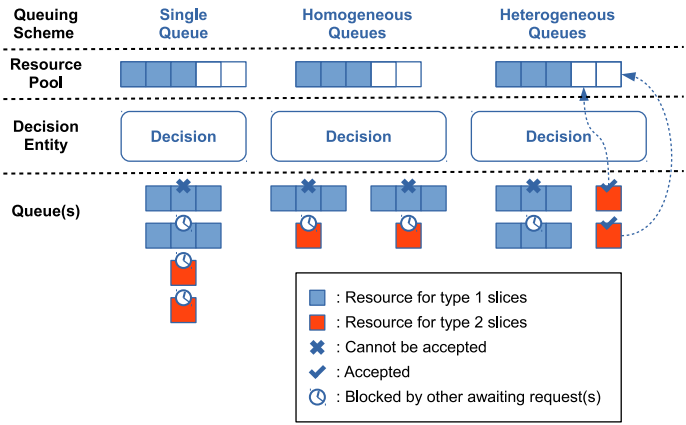}
	\caption{A simple case study on different queuing schemes.}
	\label{fig:schemes}
\end{figure}

Obviously, both the single-queue and the homogeneous multi-queue schemes can also overcome this issue by introducing a ``queue-jumping'' mechanism. However, this may require an extra design of (more complex) logic that automatically (and dynamically) decides which request is allowed to jump in the queue(s). Therefore, in this study we consider the scheme with $N$ FCFS heterogeneous queues.
%\subsection{Random Service, Single Queue}
%\subsection{Random Service, Multiple Queues}

\section{Heterogeneous multi-queue admission control}
\label{sect:multi-queue-controller}
Based on the heterogeneous multi-queue scheme, we propose in this section a novel code to present the MNO's preference for different slice types in variable states, a multi-queue admission controller for SlaaS, and analyze its queue model.

\subsection{Slice-type preference encoder}
Differing from existing studies that do not consider queuing and the single-queue scheme, in the multi-queue scheme, the MNO may receive multiple requests for slices of different types simultaneously. Therefore, instead of making a simple binary decision of accepting or declining one request, it has to either choose one from the simultaneously arriving requests to accept while declining the rest ones, or decline all of them. Especially, with heterogeneous queues, the MNO's preference for some request queue(s) over the others implies its proclivity to some slice type(s) against the others.

For an MNO that offers \change{$N$} different slice types to tenants for request, we can encode an arbitrary preference of the MNO into a \emph{preference vector} of length $N+1$:
\begin{equation}
	\Phi=[\varphi_1,\varphi_2,\dots,\varphi_{N+1}],
\end{equation}
which is a permutation of  $\{0,1,2,\dots,N\}$. The earlier a queue number $1\le n\le N$ occurs in $\Phi$, the more likely the MNO prefers slice type $n$ over the others. Note that $n=0$ denotes reserving resource for potential opportunities in future, so that all requests in the queues with values occurring in $\Phi$ after $0$ will not be served by the MNO at all. 

While being in states on (or close to) the border of space of resource feasibility $\mathbf{s}\in\mathbb{S}-\mathbb{A}$, the MNO cannot accept further request from any queue, hence the preference does not make any impact. Thus, we focus on the admissibility region $\mathbb{A}$ and assume that the MNO's preference is consistent and depends only on its current state $\mathbf{s}\in\mathbb{A}$. Thus, we can characterize the MNO's admission strategy with a $(N+1)\times\vert \mathbb{A}\vert$ \emph{preference matrix} as the following
\begin{equation}\label{equ:phi_matrix}
\begin{split}
	\mathbf{\Phi}&=[\Phi_1,\Phi_2,\dots,\Phi_{\vert\mathbb{A}\vert}]\\
	&=\begin{bmatrix}
	\phi_{1,1}&\phi_{1,2}&\dots&\phi_{1,\vert\mathbb{A}\vert}\\
	\phi_{2,1}&\phi_{2,2}&\dots&\phi_{2,\vert\mathbb{A}\vert}\\
	\vdots&\vdots&\ddots&\vdots\\
	\phi_{N+1,1}&\phi_{N+1,2}&\dots&\phi_{N+1,\vert\mathbb{A}\vert}\\
	\end{bmatrix},
\end{split}
\end{equation}
where each column $\Phi_i$ represents the MNO's preference for different slice types in a specific feasible state in $\mathbb{A}$. 

\subsection{Mechanism overview}\label{subsec:overall_mechanism}
%Let the bijection $\mathbb{A}\to\{1,2,\dots,\vert\mathbb{A}\vert\}$ denoted by $I=I_\mathbb{A}(\mathbf{s})$ 
Let $l_n$ denote the length of the $n^\text{th}$ queue, the decision entity executes the algorithm described in Fig. \ref{fig:algorithm}. The MNO keeps waiting for incoming tenant issues and responses to them upon issue arrivals. If the tenant issues to release a slice of its own, the MNO always releases it. If the tenant requests for a new slice, the request will be pushed into the corresponding queue with respect to the type of requested slice. After responding to the issue, the MNO will recursively serve the request queues in a sequence determined by its admission strategy and active slice set, until no more waiting request can be accepted. Then it stops serving the queues and waits for the next tenant issue.

\begin{figure}[!h]
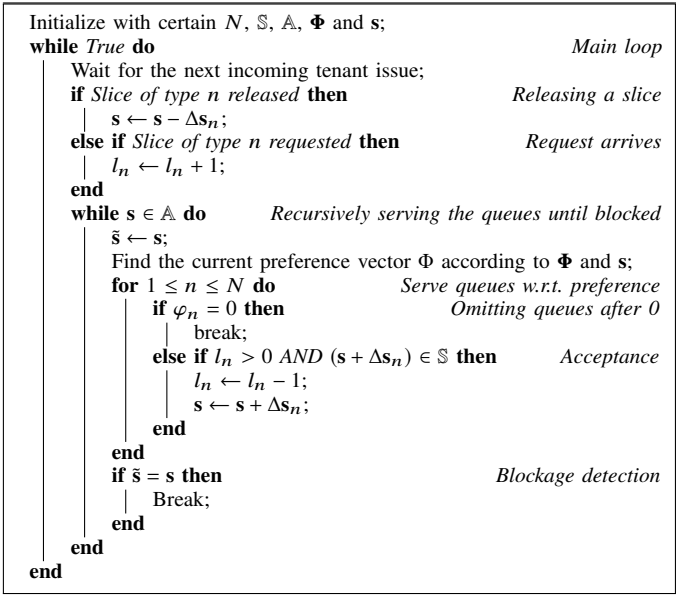

	\removelatexerror
	%	\begin{algorithmic}[1]
	\begin{algorithm}[H]
		\footnotesize
		Initialize with certain $N$, $\mathbb{S}$, $\mathbb{A}$, $\mathbf{\Phi}$ and $\mathbf{s}$\;
		\While(\hfill\emph{Main loop}){True}{
			Wait for the next incoming tenant issue\;
%			$i\gets I_\mathbb{A}(\mathbf{s})$\qquad\emph{For convenience of denotation}\;
			\uIf(\hfill\emph{Releasing a slice}){Slice of type $n$ released}{$\mathbf{s}\gets\mathbf{s}-\Delta\mathbf{s}_n$\;}
			\ElseIf(\hfill\emph{Request arrives}){Slice of type $n$ requested}{
				$l_{n}\gets l_{n}+1$\;				
			}
			\While(\hfill\emph{Recursively serving the queues until blocked}){$\mathbf{s}\in\mathbb{A}$}{
				$\tilde{\mathbf{s}}\gets\mathbf{s}$\;
                Find the current preference vector $\Phi$ according to $\mathbf{\Phi}$ and $\mathbf{s}$\;
				\For(\hfill\emph{Serve queues w.r.t. preference}){$1\le n\le N$}{
					\uIf(\hfill\emph{Omitting queues after 0}){$\varphi_{n}=0$}{break\;}
					\ElseIf(\hfill\emph{Acceptance}){$l_{n}>0$ AND $\left(\mathbf{s}+\Delta\mathbf{s}_{n}\right)\in\mathbb{S}$}{
						$l_n\gets l_n-1$\;
						$\mathbf{s}\gets\mathbf{s}+\Delta\mathbf{s}_n$\;
					}
				}
				\If(\hfill\emph{Blockage detection}){$\tilde{\mathbf{s}}=\mathbf{s}$}{Break\;}
			}
		}	
	\end{algorithm}
% 	\begin{algorithm}[H]
% 		\footnotesize
% 		Initialize with certain $N$, $\mathbb{A}$, $\mathbf{\Phi}$ and $\mathbf{s}$\;
% 		\While(\qquad\emph{Main loop}){True}{
% 			Wait for the next incoming request or slice release\;
% 			$i\gets I_\mathbb{A}(\mathbf{s})$\qquad\emph{For convenience of denotation}\;
% 			\uIf(\qquad\emph{Releasing a slice}){Slice of type $n$ released}{$\mathbf{s}\gets\mathbf{s}-\Delta\mathbf{s}_n$\;}
% 			\ElseIf(\qquad\emph{Request arrives}){Slice of type $\phi_{i,n}$ requested}{
% 				$l_{\phi_{i,n}}\gets l_{\phi_{i,n}}+1$\;
% 			}
% 			\While(\qquad\emph{Recursively serving the queues until blocked}){True}{
% 				$\tilde{\mathbf{s}}\gets\mathbf{s}$\;
% 				\For(\qquad\emph{Serve queues w.r.t. preference}){$1\le n\le N$}{
% 					$i\gets I_\mathbb{A}(\mathbf{s})$\;
% 					\uIf(\qquad\emph{Omitting queues after 0}){$\phi_{i,n}=0$}{break\;}
% 					\ElseIf{$l_{\phi_{i,n}}>0$ and $\left(\mathbf{s}+\Delta\mathbf{s}_{\phi_{i,n}}\right)\in\mathbb{A}$}{
% 						$l_{\phi_{i,n}}\gets l_{\phi_{i,n}}-1$\;
% 						$\mathbf{s}\gets\mathbf{s}+\Delta\mathbf{s}_{\phi_{i,n}}$\;
% 					}
% 				}
% 				\If(\qquad\emph{Block detection}){$\tilde{\mathbf{s}}=\mathbf{s}$}{Break\;}
% 			}
% 		}	
% 	\end{algorithm}
	\caption{The multi-queue slice admission controlling algorithm.}
	\label{fig:algorithm}
\end{figure}

\section{Network slicing controller design}
\label{sect:controller}
We analyze different characteristics of the conventional queuing models, highlighting the novel features applied to our model while designing the network slicing controller. This helps to shed the light on the main advantages and limitations of our novel admission control model.

\subsection{Analysis of inter-acceptance time}
We consider request arrivals of every slice type as an independent Poisson process, so that the inter-arrival time between requests in every queue is an independent exponential random process. Conversely, the request acceptance rate of every queue is jointly determined by the slice releases of all types, and the MNO's preference strategy.

\begin{theorem}\label{theorem:geom_iat}
	Consider a heterogeneous multi-queue slice admission controller that executes the algorithm in Fig. \ref{fig:algorithm} with a consistent preference matrix. The acceptance in different queues are mutually independent Poisson processes, if:	
	 \begin{enumerate*}
		\item the arrivals of new requests and releases of active slices are mutually independent Poisson processes for every individual slice type;
		\item the arrivals of different slice types are mutually independent from each other, the releases of different slice types are mutually independent from each other.
	\end{enumerate*}
	
\end{theorem}

\begin{proof}
	First, extend the system (MNO) state $\mathbf{s}$ with all queue lengths to obtain the \emph{controller state} $\hat{\mathbf{s}}=[\mathbf{s},l_1,l_2,\dots,l_N]$, and therefore the infinite discrete domain $\hat{\mathbb{A}}=\mathbb{A}\times \mathbb{N}^N$. Let the bijection $\mathbb{A}\leftrightarrow\{1,2,\dots,\vert\mathbb{A}\vert\}$ denoted by $I=I_\mathbb{A}(\mathbf{s})$, we call $\hat{\mathbf{s}}\in\hat{\mathbb{A}}$ a \emph{transient} state if $\exists n\in\mathcal{N}$ such that:
	\begin{empheq}[left=\empheqlbrace]{align}
		&\phi_{I,k}\neq 0,\quad\forall k<n;\\
		&l_n>0;\\
		&\left(\mathbf{s}+\Delta\mathbf{s}_{\phi_{n,I}}\right)\in\mathbb{A}.
	\end{empheq}
	Otherwise, we call $\hat{\mathbf{s}}$ a $\emph{steady}$ state.
	According to the algorithm in Fig. \ref{fig:algorithm}, when the controller is in a transient state it always accepts a request in its queues immediately and therefore keeps jumping to another state until it reaches a steady state. Every transient state leads to one and only one certain steady state. On the other hand, the controller can reasonably (but not always) leave a steady state only when a new request arrives or a slice is released.
	
Thus, given a certain sequence of request arriving and slice releasing events in the next period, we can obtain the transition path of the controller state, and therewith determine whether the first awaiting request in an arbitrary queue will be accepted during that period. Denote the time that the first awaiting request in the $n^\text{th}$ queue still has to wait until it is accepted as $t_{\text{w},n}$, it yields that
	\begin{equation}
		\text{Prob}(t_{\text{w},n}>T)=1-\prod\limits_{\mathbf{e}\in\mathbb{E}_n}\text{Prob}(Arr(T)=\mathbf{e}),
	\end{equation}
	where $\mathbb{E}_n$ is the set of all event sequences that can lead to an acceptance of request in the $n^\text{th}$ queue, and $Arr(T)$ denotes the event sequence arriving in the next period of $T$. As the request arrivals and releases of different slice types are mutually independent Poisson processes, we know that all $\mathbf{e}\in\mathbb{E}_T$ are also approximately Poissonian (proven as a feature of \emph{dependent trials} \cite{parzen1960modern,chen1975poisson}). Thus, due to the Markovian behavior of Poisson processes, we can write the following
	\begin{equation}
	\begin{split}
		\text{Prob}(Arr(T)=\mathbf{e})=\text{Prob}(Arr(T+t)=\mathbf{e}~\vert~Arr(t)\neq\mathbf{e})\\
		\forall[\mathbf{e},t,T]\in\left(\mathbb{E}_T\times\mathbb{N}^2\right),
	\end{split}
	\end{equation}
	and thus
	\begin{equation}
	\begin{split}
		&\text{Prob}(t_{\text{w},n}>T+t)\\
		=&1-\prod\limits_{\mathbf{e}\in\mathbb{E}_n}\text{Prob}(Arr(T+t)=\mathbf{e}~\vert~Arr(t)\neq\mathbf{e})\\
		=&\text{Prob}(t_{\text{w},n}>T+t~\vert~t_{\text{w},n}>t),\quad\forall[t,T]\in\mathbb{N}^2.
	\end{split}
	\label{equ:memoryless_waiting_time}
	\end{equation}
Eq.~\eqref{equ:memoryless_waiting_time} implies that the remaining waiting time for acceptance of the first request in queue $n$ is \emph{memoryless}. Due to the fact that the only two classes of memoryless distributions are exponential (continuous) and geometric (discrete) distributions, we can assert that the request acceptance in every queue is a Poisson process.
\end{proof}

\subsection{Queuing-theoretic analysis}
While considering both request arrivals and request acceptances (service) as Poisson processes, every request queue is a classic $\text{M}/\text{M}/1$ queuing system, known as single-server birth-death system~\cite{shortle2018fundamentals}. Hence, many features of birth-death model can be directly applied.
\subsubsection{Little's Formula}
For slice type (queue) $n$, given its request arrival rate $\lambda_n$, according to the famous Little's formula\cite{little1961proof} there is
\begin{equation}\label{equ:little}
	L_n=\lambda_n\overline{W}_n,
\end{equation}
where $L_n$ and $\overline{W}_n$ represent the mean length of queue $n$ and the average waiting time in queue $n$, respectively.

\subsubsection{Steady Queue State Probability}
Given the request arrival rate $\lambda_n$ and acceptance rate $\mu_n$ of queue $n$, the probability that the queue steadily consists of $l$ requests at an arbitrary time instant is geometrically distributed, i.e.,
\begin{equation}\label{equ:steady_queue_state}
	p_n(l)=(1-\rho)\rho^l,
\end{equation}
where $\rho_n=\lambda_n/\mu_n<1$ is the \emph{work load rate} of queue $n$.

\subsubsection{Waiting Time Distribution}
The probability density function (PDF) of an arbitrary type-$n$ request's waiting time is
\begin{equation}\label{equ:pdf_waiting_time}
	f(W_n)=\begin{cases}
		0& W_n<0\\
		(\mu_n-\lambda_n)e^{-(\mu_n-\lambda_n)W_n}&W_n\ge 0
	\end{cases},
\end{equation}
and the cumulative density function (CDF) is
\begin{equation}\label{equ:cdf_waiting_time}
	F(W_n)=\begin{cases}
		0& W_n<0\\
		1-e^{-(\mu_n-\lambda_n)W_n}&W_n\ge 0
	\end{cases}.
\end{equation}

\subsection{Extension: impatient tenants}
From Eqs. (\ref{equ:little}--\ref{equ:cdf_waiting_time}) it is clear that both $L_n$ and $W_n$ converge only when $\lambda_n<\mu_n$. Otherwise, when the request acceptance rate is lower than the arrival rate in queue $n$, the queue length will infinitely increase, and therefore also the mean waiting time. This is known as the necessary and sufficient condition of statistical equilibrium in queuing processes, as stated and proven by \emph{Kendall} in work~\cite{kendall1951some}.

However, in a real slice admission controller, there are various situations where $\lambda_n\ge\mu_n$ for some $n$, including cases
\begin{itemize}
	\item when the controller is specified with an inappropriate strategy, so that requests in the queue $n$ is rarely or even never accepted despite of resource feasibility;
	\item when the release rates of active slices are low, so that the resource pool fails to support a sufficiently high $\mu_n$ regardless of any admission strategy.
\end{itemize}
There are two mechanisms that prevent queuing systems from such divergence.
On the one hand, the system may force to truncate a queue at some maximal length, and forbid this queue to take any new request before it is shortened. On the other hand, the clients may lose patience while waiting, and leave the queues before being served (e.g., for looking for some other MNO with resource availability). In the scenario of SlaaS, the system (MNO) is probably very cautious with refusing requests,  while the waiting time can be critical to the customers (tenants). Therefore, here we consider no queue truncation but queues with impatience.

Usually, impatience in queues can occur in three different behaviors:
$i$) balking, i.e. customers being reluctant to join a queue upon arrival, $ii$) reneging, i.e. customers leaving the queue after joining and waiting, and $iii$) jockeying from long lines to shorter ones. As the heterogeneous multi-queue design disables jockeying, here we consider the balking and reneging phenomena.

\noindent{\bf Balking Model.} The phenomenon of balking can be modeled in such a way, that every arrival request of slice type $n$ enters the queue with a probability $b_n$, which is a monotonically decreasing function of the current queue length $l_n$.
\emph{Ancker} and \emph{Gafarian} have proposed two different balking models in~\cite{ancker1963some1,ancker1963some2}. The first model considers a linear balking factor $1-b_n=l_n/l_{n,\max{}}$, where $l_{n,\max{}}$ is the upper bound of $l_n$ for queue truncation. The second one considers a non-linear balking factor as follows
\begin{equation}\label{equ:inv_prop_balking_model}
	1-b_n=\begin{cases}
		0&l_n=0\\1-\beta_n/l_n&l_n\in\mathbb{N}^+
	\end{cases},
\end{equation}
where $\beta_n\in[0,1]$ measures the willingness of tenants requesting type-$n$ slices to wait. In cases that the tenant has knowledge about $\mu_n$, \emph{Shortle} et al. suggest another non-linear balking model $1-b_n=1-e^{-\beta_n l_n/\mu_n}$ where $\beta_n>0$ \cite{shortle2018fundamentals}. Here we consider the hyperbolic balking model described by Eq.~\eqref{equ:inv_prop_balking_model}.

\noindent{\bf Reneging  Model.} The phenomenon of reneging can be modeled by randomly assigning an individual maximal waiting time to every request when it joins the queue. The request will leave the queue after that maximal waiting time if it has not been accepted yet. Following \emph{Ancker} and \emph{Gafarian} \cite{ancker1963some2}, we consider the maximal waiting time for every type-$n$ request as an exponential random variable $W_{\max{},n}\sim\text{Exp}(\alpha_n)$, where $1/\alpha_n>0$ is the mean maximal waiting time in queue $n$.

\subsection{Performances with balking and reneging}
It should be noted that the balking and reneging processes are with memory, leading to a non-Markovian behavior of request acceptances. However, under low balking and reneging rates, this impact can be negligible and the acceptance process can still be approximated as Poissonian. When the balking and reneging rates rise to significant levels, the memory of acceptance process shall be considered, as demonstrated in Section~\ref{subsec:geometric_iat} by means of simulations.

Under a combination of hyperbolic balking and exponential reneging, the steady state probability of having $l$ requests in the queue $n$ is 
\begin{equation}
	p_n(l)=\begin{cases}
		\frac{1}{1+(\delta_n)^{1-\gamma_n/2}[\Gamma(\gamma_n)/\beta_n]I_{\gamma_n}(2\sqrt{\delta_n})}&l=0\\\\
		\frac{\delta_n^lp_n(0)}{\beta_n(l-1)!\prod_{j=0}^{l-1}(\gamma_n+j)}&l\in\mathbb{N}^+
	\end{cases},
\end{equation}
where $\gamma_n=\mu_n/\alpha_n$, $\delta_n=\lambda_n\beta_n/\alpha_n$,  $I_{\gamma_n}(\cdot)$ is the modified Bessel's function of the first kind and order $\gamma_n$.

Meanwhile, we are interested in three different distributions of waiting time spent in a queue $n$: $i$) $f_\text{a}(W_n)$ for requests that are eventually accepted, $ii$) $f_\text{r}(W_n)$ for requests that renege and $iii$) $f_\text{q}(W_n)$ for all requests that join the queue.
Let us define $A_n$ and $J_n$ as the events of request being accepted and joining the queue $n$, respectively. There are 
\begin{align}
P(A_n)&=\frac{[1-p_n(0)]\beta_n\gamma_n}{\delta_n},\\
P(A_n,J_n)&=\frac{[1-p_n(0)]\beta_n\gamma_n}{\delta_n}-p_n(0),\\
P(A_n\vert J_n)&=\frac{\Gamma(\gamma_n+1)I_{\gamma_n}\left(2\sqrt{\delta_n}\right)-\left(\sqrt{\delta_n}\right)^{\gamma_n}}{\sqrt{\delta_n}\Gamma(\gamma_n)I_{\gamma_n-1}\left(2\sqrt{\delta_n}\right)-\left(\sqrt{\delta_n}\right)^{\gamma_n}}.
\end{align}
It can be obtained that
{%\scriptsize
\begin{align}
	f_\text{a}(W_n)&=\frac{p_n(0)\lambda_n\beta_ne^{-(\mu_n+\alpha_n)W_n}I_1\left[2\sqrt{\delta_n(1-e^{-\alpha_nW_n})}\right]}{P(A_n,J_n)\sqrt{\delta_n\left(1-e^{-\alpha_nW_n}\right)}},\\
	f_\text{r}(W_n)&=\alpha_ne^{-\alpha_nW_n}\frac{1-P(A_n\vert J_n)g(W_n)}{1-P(A_n\vert J_n))},\\
	f_\text{q}(W_n)&=P(A_n\vert J_n)\left[f_\text{a}(W_n)-\alpha_ne^{-\alpha_nW_n}g(W_n)\right]+\alpha_ne^{-\alpha_nW_n},
\end{align}}
where $g(W_n)=\int_{0}^{W_n}e^{\alpha_n\xi}f_\text{a}(\xi)\text{d}\xi$.

The expectations of waiting times are therefore
\begin{align}
\overline{W}_{\text{a},n}&\frac{p_n(0)}{P(A_n,W_n)}\sum_{i=1}^{+\infty}\left[\frac{\delta_n^i}{i!}\prod_{j=1}^{n}(\gamma_n+j)\right]\sum_{k=1}^{i}\frac{1}{\gamma_n+k},\\
\overline{W}_{\text{r},n}&=\frac{1}{\alpha_n}-\frac{P(A_n\vert W_n)\overline{W}_{\text{q},n}}{1-P(A_n\vert W_n)},\\
\overline{W}_{\text{q},n}&=\frac{1-P(A_n\vert W_n)}{\alpha_n}.
\end{align}

\section{Strategy optimization}
\label{sect:optimization}
In slice admission control, there are various performance metrics that may include: the overall network utility rate, the admission rate and the average request waiting time.

The network utility of a slice can be differently defined, such as the periodical payment that the MNO receives from the tenant, or the generated network throughput, etc. It is common to consider the utility rate of a slice as determined by the slice type, and the overall network utility rate at any time instant $t$ as the sum of utility rates of all slices under maintenance:
\begin{equation}\label{equ:instant_net_utility_rate}
	u_\Sigma(t)=\sum_{n=1}^{N}s_n(t)u_n,
\end{equation}
where $s_n(t)$ is the number of type-$n$ slices under maintenance at time $t$, and $u_n$ is the utility rate of every type-$n$ slice. In long term, the average overall network utility rate can be estimated from the acceptance and releasing rates of different slice types:
\begin{equation}\label{equ:mean_net_utility_rate}
	\overline{u}_\Sigma=\sum_{n=1}^{N}\frac{\mu_n u_n }{\eta_n},
\end{equation}
where $\eta_n$ is the releasing rate per type-$n$ slice.

The average waiting time of all requests in queues is
\begin{equation}
	\overline{W}_\text{q}=\frac{\sum\limits_{n=1}^{N}\overline{W}_{\text{q},n}L_n}{\sum\limits_{n=1}^{N}L_n}.
\end{equation}
The overall admission rate is the following
\begin{equation}
	\overline{P}(A)=\frac{\sum_{n=1}^N\lambda_nP(A_n)}{\sum_{n=1}^N\lambda_n}.
\end{equation}

All three criteria are determined by the request behavior parameters $\alpha_n,\beta_n,\lambda_n$ and the acceptance rate $\mu_n$. Given a certain combination of $[\alpha_n,\beta_n,\lambda_n,\eta_n]$, where $1/\eta_n$ is the average lifetime of type $n$ slices, $\mu_n$ is uniquely determined by the MNO's strategy, i.e. by the preference matrix $\mathbf{\Phi}$. Hence, with consistent behaviors of request arrival and slice releasing, we can optimize either of them by selecting the best $\mathbf{\Phi}$.

A major challenge for analysis exists in the complex relation between the acceptance rates $[\mu_1,\mu_2,\dots,\mu_N]$  and the strategy $\mathbf{\Phi}$, as $\mathbf{\Phi}$ does not directly imply the MNO's action or statistics, but only its preference. 

Nevertheless, if the steady-state probability of queue lengths $p_n(l)$, as defined in Eq.~\eqref{equ:steady_queue_state}, is known or measurable for all $n\in\mathcal{N}$, we can estimate $\mu_n$ for all $n$ with respect to $\mathbf{\Phi}$ and the initial state $\mathbf{s}_\text{init}$ as follows.

First, define a bijection $\mathbb{S}\leftrightarrow\{1,2,\dots,\vert\mathbb{S}\vert\}$ as $J=J_\mathbb{S}(\mathbf{s})$ where $J_\mathbb{S}(\mathbf{s})=I_\mathbb{A}(\mathbf{s})$ for all $\mathbf{s}\in\mathbb{A}$. Then extend the definitions in Eqs. (\ref{equ:slice_incremental_vector}), (\ref{equ:phi_matrix}) and (\ref{equ:steady_queue_state}) with
\begin{align}
	&\Delta\mathbf{s}_0=\underbrace{[0,0,\dots,0]}_N,\\
	&\tilde\phi_{i,j}=\begin{cases}
		0&j>\vert\mathbb{A}\vert\\
		\phi_{i,j}&j\le\vert\mathbb{A}\vert
	\end{cases},\forall i\in\{1,2,\dots,N+1\},\\
	&p_{0}(0)=0,
\end{align} 
respectively. The probability of state transition from any $\mathbf{s}\in\mathbb{S}$ to $\mathbf{s}+\Delta\mathbf{s}$ can be then calculated as
\begin{equation}
\text{Prob}(\mathbf{s}\to\mathbf{s}+\Delta\mathbf{s}_n)=\prod_{k=1}^{n-1}p_{\tilde\phi_{k,J}}(0)(1-p_{\tilde\phi_{n,J}}(0)).
\end{equation}

Thus, when the initial state $\mathbf{s}_\text{init}$ is known, we can obtain the long-term probability distribution of system state $\mathbf{s}$ as
\begin{equation}
	\text{Prob}(\mathbf{s}_j~\vert~\mathbf{s}_\text{init}=\mathbf{s}_i)=\lim\limits_{K\to\infty}\frac{1}{K}\sum_{k=0}^{K}[\mathbf{\Psi}^k]_{i,j},
\end{equation}
where $\mathbf{\Psi}$ is the transition matrix:
\begin{equation}
	\mathbf{\Psi}=\begin{bmatrix}
	\Psi_{1,1} &\Psi_{1,2} & \dots & \Psi_{1,\vert\mathbb{S}\vert}\\
	\Psi_{2,1} &\Psi_{2,2} & \dots & \Psi_{2,\vert\mathbb{S}\vert}\\
	\vdots&\vdots&\ddots&\vdots\\
	\Psi_{\vert\mathbb{S}\vert,1} &\Psi_{\vert\mathbb{S}\vert,2} & \dots & \Psi_{\vert\mathbb{S}\vert,\vert\mathbb{S}\vert}\\
	\end{bmatrix},
\end{equation}
and $\Psi_{i,j}=\text{Prob}(\mathbf{s}_i\to\mathbf{s}_j)$. 

More generally, if not the exact value but the  probability distribution of  the initial state is available as $P_{\text{init}}=[p_{\text{init}}(\mathbf{s}_1),p_{\text{init}}(\mathbf{s}_2),\dots,p_{\text{init}}(\mathbf{s}_{\vert\mathbb{S}\vert})]$, the long-term probability distribution $\mathbf{s}$ is the following
\begin{equation}
\text{Prob}(\mathbf{s}_j~\vert~P_\text{init})=\lim\limits_{K\to\infty}\frac{1}{K}\sum_{k=0}^{K}\sum_{i=1}^{\vert\mathbb{S}\vert}p_{\text{init}}(\mathbf{s}_j)[\mathbf{\Psi}^k]_{i,j}.
\end{equation}

We can obtain the expected active slice number $\overline{s}_n$ of every slice type $n$ as a function of $\mathbf{\Psi}$ and thus, as a function of $\mathbf{\Phi}$. Now, recalling Eqs.~(\ref{equ:instant_net_utility_rate}--\ref{equ:mean_net_utility_rate}) it yields that
\begin{equation}
	\overline{s}_n=\frac{\mu_n}{\eta_n},
\end{equation}
and then we can write the following
\begin{equation}\label{equ:cost_function}
\begin{split}
	\mu_n&=\frac{\overline{s}_n}{\eta_n}=\frac{\sum\limits_{\mathbf{s}\in\mathbb{S}}\text{Prob}(\mathbf{s}~\vert~P_\text{init})s_n}{\eta_n}\\
	&=\frac{1}{\eta_n}\sum\limits_{\mathbf{s}\in\mathbb{S}}\lim\limits_{K\to\infty}\frac{1}{K}\sum_{k=0}^{K}\sum_{i=1}^{\vert\mathbb{S}\vert}p_{\text{init}}(\mathbf{s}_j)[\mathbf{\Psi}^k]_{i,j}.
\end{split}
\end{equation}

Based on this analytical expression, we are able to optimize $[\mu_1,\mu_2,\dots,\mu_n]$ with respect to $\mathbf{\Phi}$. However, it is evident that Eq.~\eqref{equ:cost_function} is non-convex w.r.t. $\mathbf{\Phi}$, which prohibits analytical solution of the global optimum. On the other hand, the overall domain size of $\mathbf{\Phi}$ is $2^{(N+1)\vert\mathbb{A}\vert}$, which can assume unaffordable high values for any realistic dimension of $\vert\mathbb{A}\vert$ in practical networks, making the exhaustive search impossible. This is an integer linear programming (ILP) problem that is proven to be NP-Hard, therefore advanced machine learning and heuristic search methods are needed to solve it with affordable efforts of computation.

%
%However, there is a major challenge that the acceptance rates $[\mu_1,\mu_2,\dots,\mu_N]$ are not convex about $\mathbf{\Phi}$, disabling analytical solution of the optimum. 
%
%\todo{\begin{itemize}
%\item Both $P(A_n)$ and $\mu_n$ are determined by $\mathbf{\Phi}$ and tightly coupled with each other
%\item Assume we have the steady-state probability of queue lengths $p_n(l)$ for all $n\in\{1,2,\dots,N\}$ as defined in (\ref{equ:steady_queue_state}), we can calculate the state transition probability from $\Phi_i$ where $I=I_\mathbb{A}(\mathbf{s})$ is defined in Sec. \ref{subsec:overall_mechanism} by
%\begin{equation}
%	\text{Prob}(\mathbf{s}_i\to\mathbf{s}_i+\Delta\mathbf{s}_n)=\prod_{k=1}^{n-1}p_{\phi_{k,i}}(0)(1-p_{\phi_{n,i}}(0)),
%\end{equation}
%where we define $p_{0}(0)=0$ in addition to (\ref{equ:steady_queue_state}), and $\Delta\mathbf{s}_0=\underbrace{[0,0,\dots,0]}_N$ in addition to (\ref{equ:slice_incremental_vector}). Thus, we can calculate an unique transition probability matrix of the active slice set with respect to an arbitrary preference matrix $\mathbf{\Phi}$, and therewith design a heuristic optimization method
%\item Remark: from the transition prob. matrix and initial distribution or arbitrary initial state we can have the state prob. distribution of $\mathbf{s}$. Going back to Eqs. (\ref{equ:instant_net_utility_rate},\ref{equ:mean_net_utility_rate}) we know $\overline{s}_n=\frac{\mu_n}{\eta_n}$. As $\eta_n$ is known, this implies $\mu_n$ as function of  $\mathbf{\Phi}$ for all $n\in\{1,2,\dots,N\}$
%\item Try ANN?.
%\end{itemize}}

\section{Numerical simulations}
\label{sect:perf_eval}
To carry out simulations in a consistently specified environment, we consider an MNO with a two-dimensional ($M=2$) normalized resource pool $\mathbf{r}=[r_1,r_2]=[1,1]$. $N=2$ slice types are defined in two service demand scenarios, as shown in Tab. \ref{tab:slice_profiles}. Note that $\alpha_n$ and $\beta_n$ are only applicable when the simulation considers balking and reneging, respectively.
\begin{table}[!htpb]
	\centering
	\begin{tabular}{c|c|c|c|c|c|c}
		\toprule[1.5px]
		\textbf{Type ($n$)}&$\mathbf{c}_n$&$\lambda_n$&$1/\eta_n$&$u_n$&$\alpha_n$&$\beta_n$\\\hline
		\multirow{2}{*}{1}&\multirow{2}{*}{$[0.01,0.05]$}&2 (Scenario 1)&\multirow{2}{*}{5}&\multirow{2}{*}{1}&\multirow{4}{*}{1}&\multirow{4}{*}{0.02}\\\cline{3-3}
		&&6 (Scenario 2)&&&&\\\cline{1-5}
		\multirow{2}{*}{2}&\multirow{2}{*}{$[0.2,0.04]$}&0.5 (Scenario 1)&\multirow{2}{*}{2}&\multirow{2}{*}{10}&&\\\cline{3-3}
		&&1.5 (Scenario 2)&&&&\\
		\bottomrule[1.5px]
	\end{tabular}
	\caption{Specifications of two reference slice types}
	\label{tab:slice_profiles}
    \vspace{-4mm}
\end{table}

\subsection{Verification of geometric IAT distribution}\label{subsec:geometric_iat}
In case of patient tenants, Theorem \ref{theorem:geom_iat} can also be verified through numerical simulations. We take the slice specifications in scenario $1$, disable balking and reneging events, and randomly generate $500$ slicing strategies. For each strategy, $20$ rounds of Monte-Carlo tests are executed. In each testing round, an MNO with a 2-queue slice admission controller is initialized to a random but fully resource-utilized state, and then operates under the consistent strategy for $40$ operations periods. Then we investigate the distribution of inter-acceptance time (IAT) for each queue, and fit the measurements with geometric distributions, which is the discrete-time version of exponential distribution. A sample result is shown in Fig.~\ref{subfig:geometric_acceptance_sample}, where a good fitting performance can be observed.

To quantitatively evaluate the fitness, we compute the Kullback-Leibler divergence (KLD)~\cite{kullback1997information} for every strategy:
\begin{equation}
	D_\text{KL}(P_\text{IAT}~\vert~\text{Geom.})=\sum\limits_{k=0}^{\infty}p_\text{IAT}(k)\log\frac{p_\text{IAT}(k)}{(1-\hat{p})^k\hat{p}},
\end{equation}
where $p_\text{IAT}(k)$ is the empirical probability mess function (PMF) of the measured IAT, and $(1-\hat{p})^k\hat{p}$ is the geometric PMF with fitted parameter $\hat{p}$. KLD is an indicator of fitness between two distributions, which equals $0$ for two identical distributions and approaches to $1$ for two completely irrelevant distributions. The KLD distribution over all $500$ tested random strategies is depicted in Fig.~\ref{subfig:geometric_acceptance_KLD}, which shows a satisfactory fitness for both queues (slice types).
\begin{figure}[!htbp]
	\centering
	\begin{subfigure}{.41\textwidth}
		\centering
		\includegraphics[width=\textwidth]{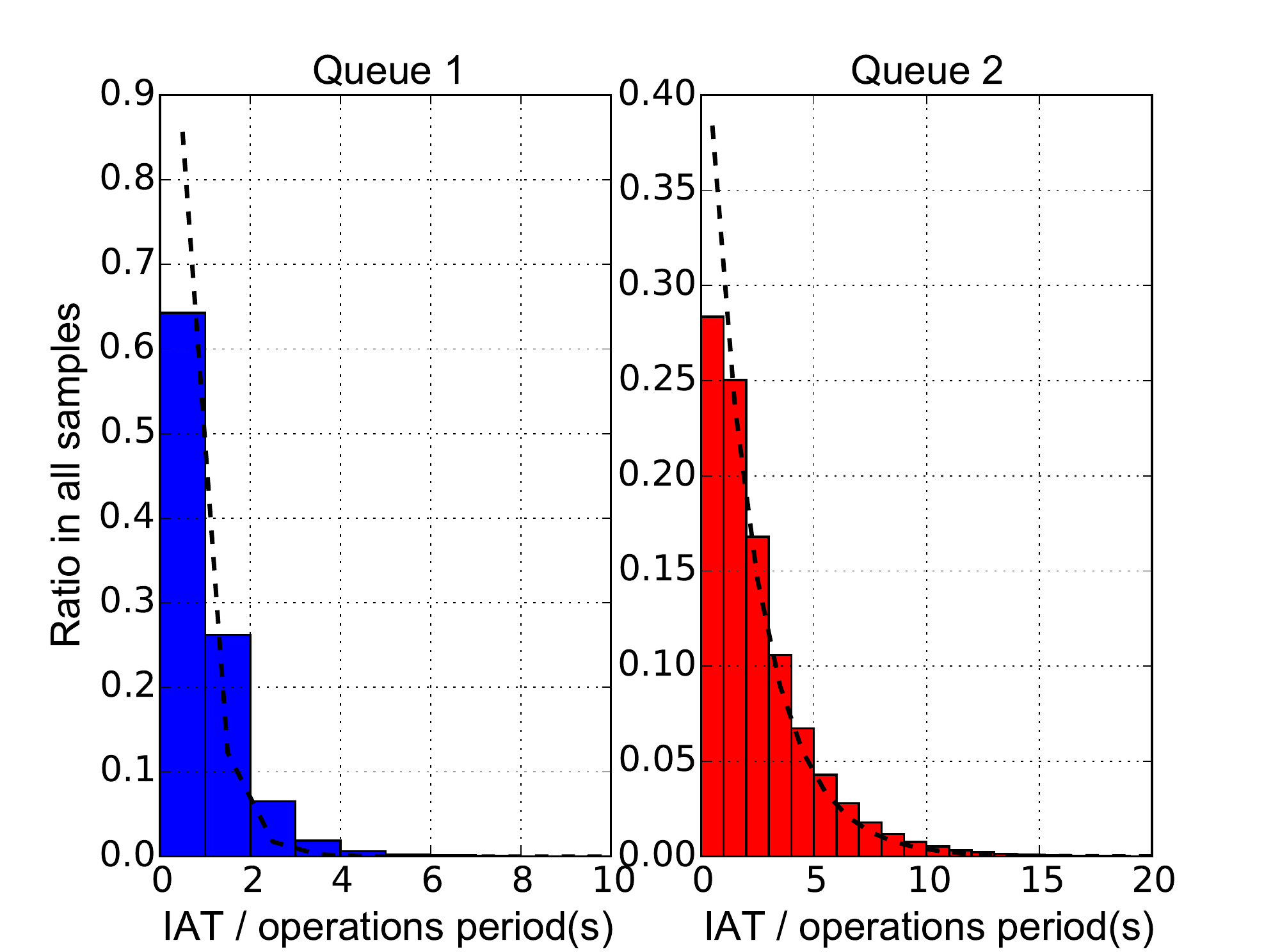}
		\caption{The distribution of inter-acceptance time in two different queues under a random strategy, fitted as geometric distribution.}
		\label{subfig:geometric_acceptance_sample}
	\end{subfigure}
	\begin{subfigure}{.4\textwidth}
		\centering
		\includegraphics[width=\textwidth]{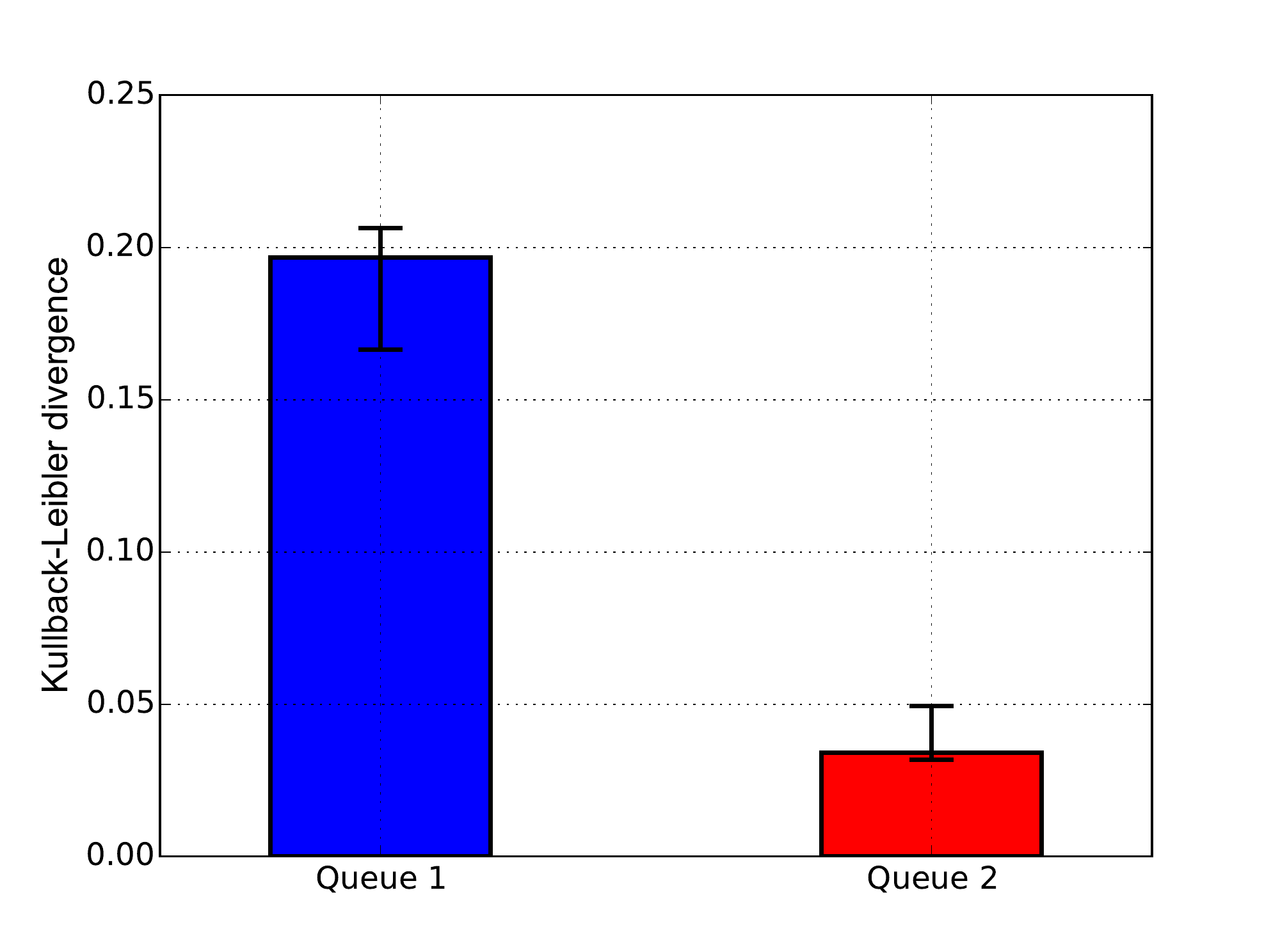}
		\caption{The Kullback-Leibler divergence of fitting the IAT distribution as geometric distribution,  500 random strategies tested.}
		\label{subfig:geometric_acceptance_KLD}
	\end{subfigure}
	\caption{The IAT of every individual queue under an arbitrary strategy is geometrically distributed.}
	\label{fig:geometric_acceptance}
\end{figure}

Furthermore, to verify the impact of impatient tenants' behavior, we activate the mechanisms of balking and reneging, and repeat the aforementioned simulation procedure in both scenarios $1$ and $2$. The results are illustrated in Fig.~\ref{fig:geometric_acceptance_balking_renaging}. Compared to the case of patient tenants, we can observe an increase of KLD in both scenarios here, especially in scenario $2$, confirming our assertion that the behaviors of balking and reneging will remove the Markovian feature of the system. However, when the balking and reneging rates are low (e.g., when the queues are short such like in scenario $1$), such impact can be slight enough to be neglected. 
\begin{figure}[tp]
	\centering
	\begin{subfigure}{.45\textwidth}
		\centering
		\includegraphics[width=\textwidth]{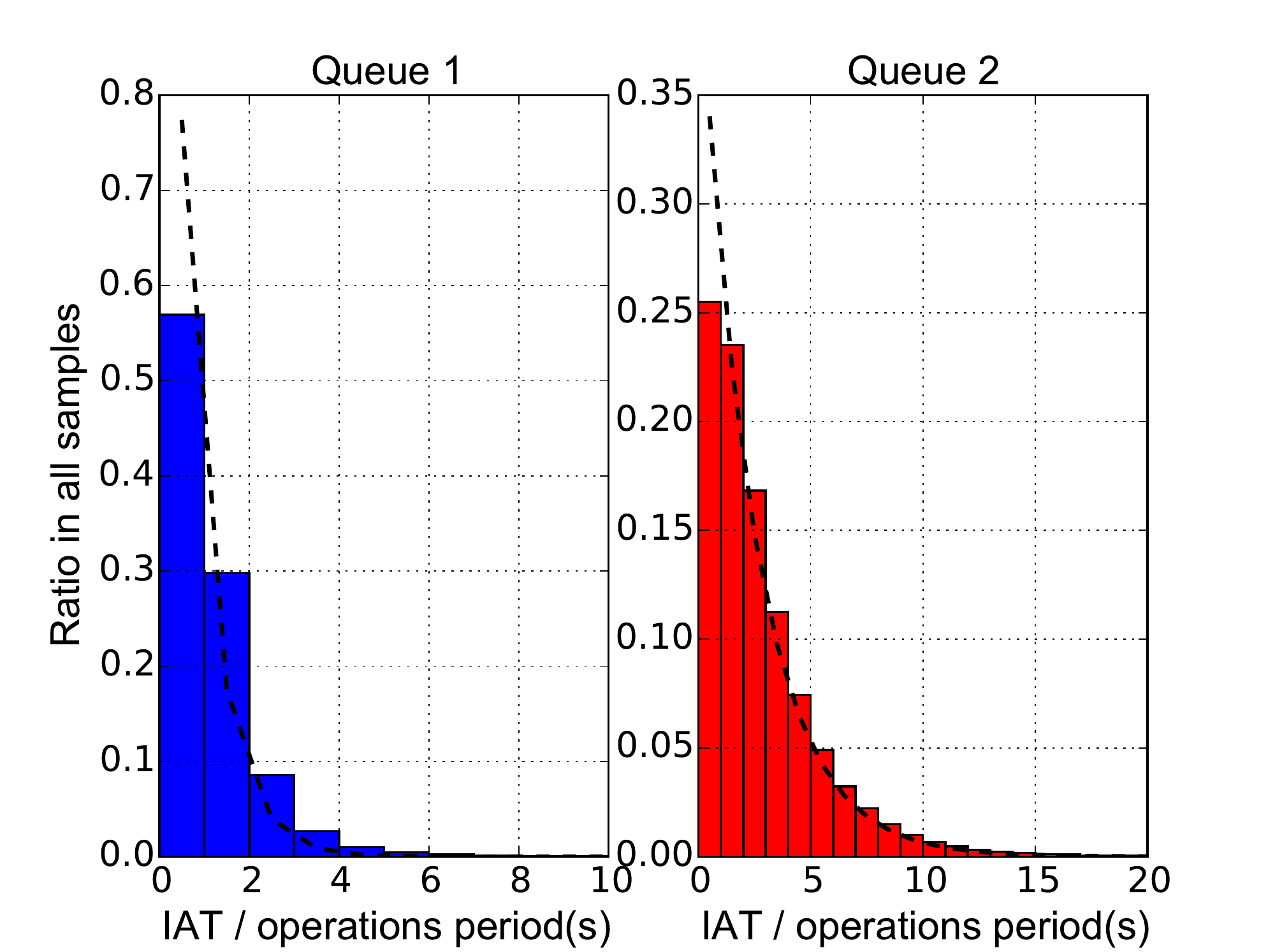}
		\caption{IAT distributions in scenario 1 under a random strategy, fitted as geometric distributions.}
	\end{subfigure}
	\begin{subfigure}{.45\textwidth}
		\centering
		\includegraphics[width=\textwidth]{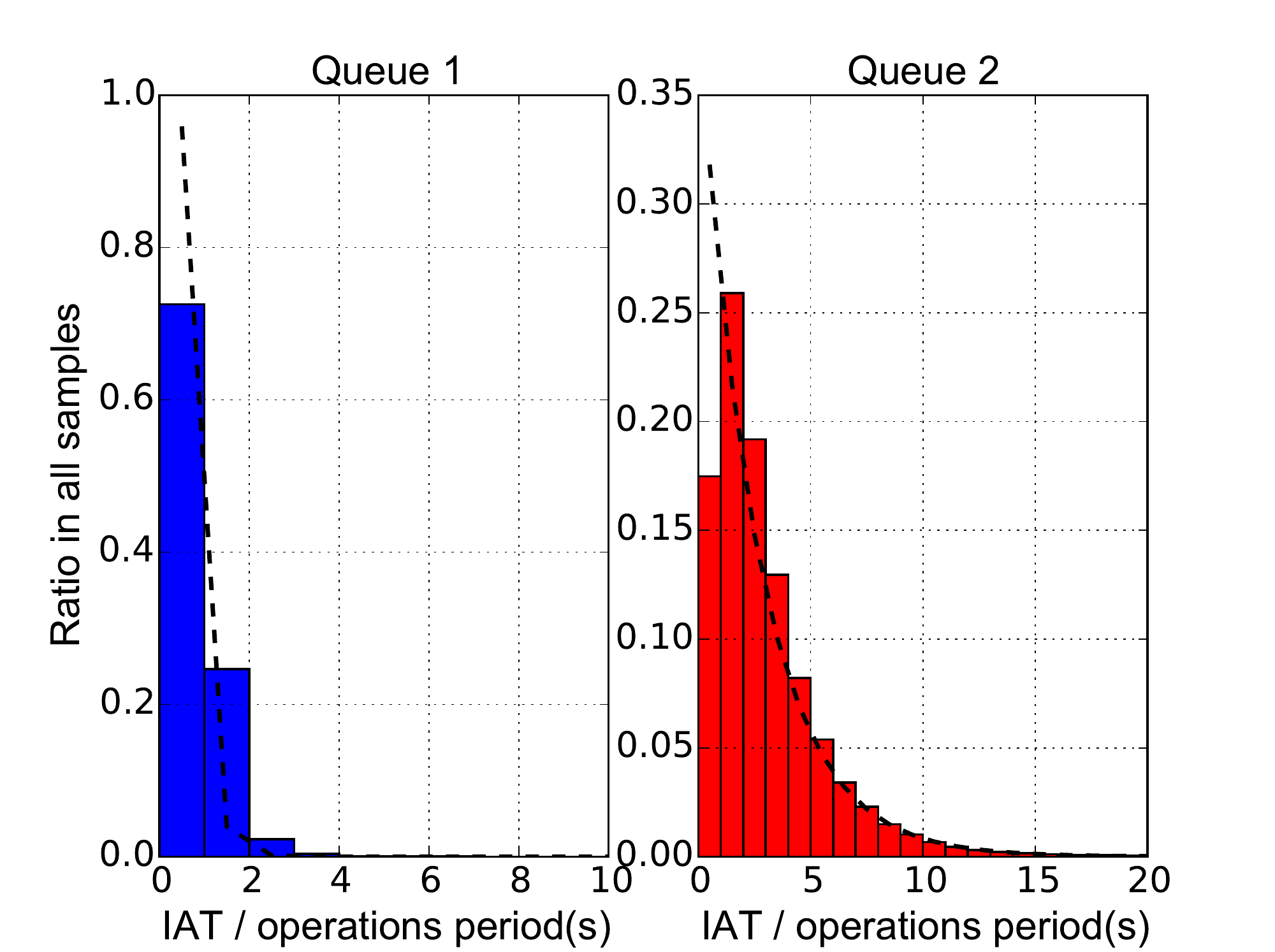}
		\caption{IAT distributions in Scenario 2 under a random strategy, fitted as geometric distributions.}
	\end{subfigure}
	\begin{subfigure}{.43\textwidth}
		\centering
		\includegraphics[width=\textwidth]{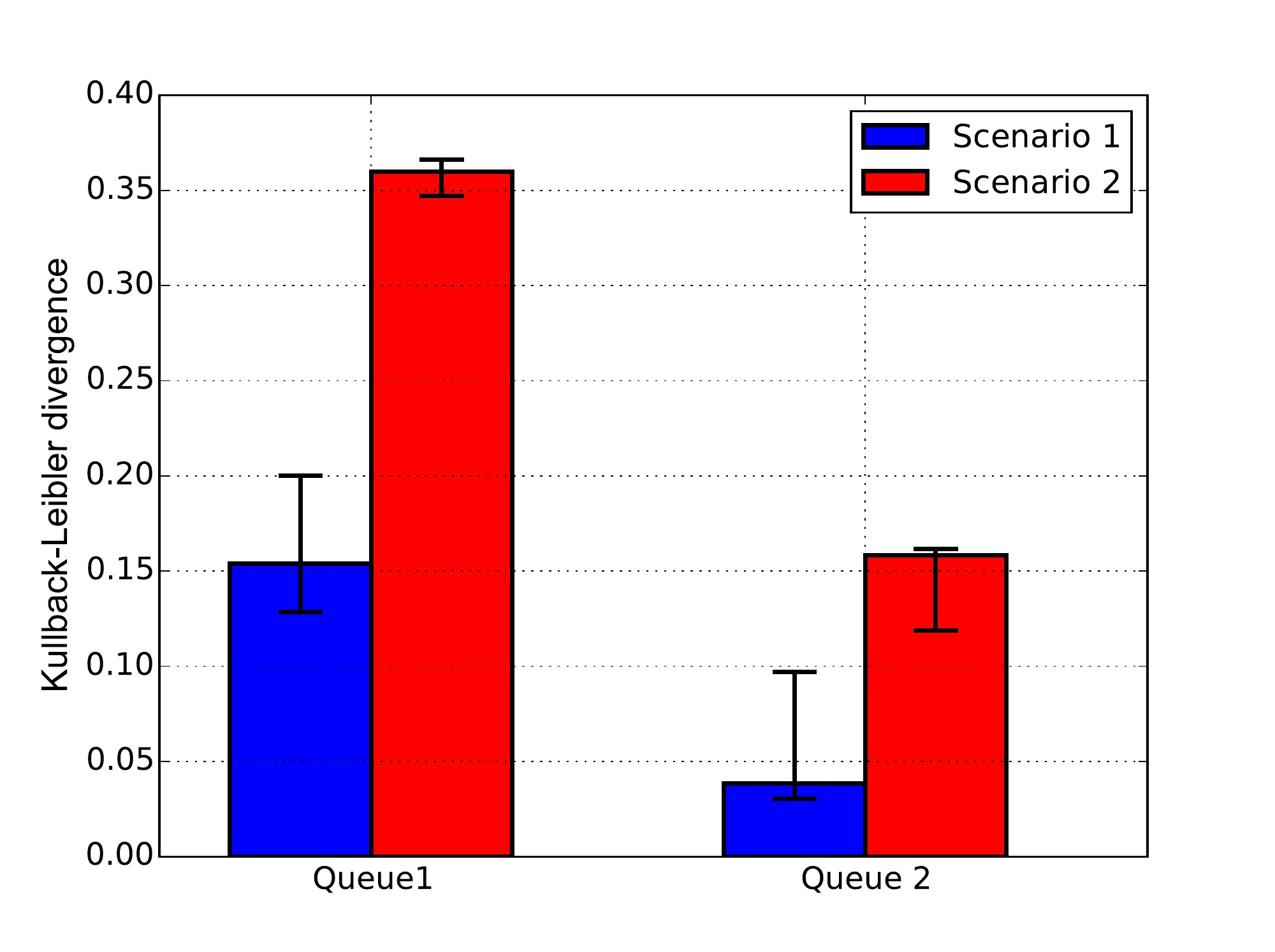}
		\caption{The KLD of fitting the IAT distribution as geometric distribution, in different scenarios.}
	\end{subfigure}
	\caption{Balking and renaging lead to non-Poisson admissions, the impact increases with the balking and renaging rates.}
	\label{fig:geometric_acceptance_balking_renaging}
\end{figure}

\subsection{Evaluation of the proposed controller}
To verify the effectiveness and potential in optimization of the proposed multi-queue slice admission controlling mechanism, we generate $10~000$ random strategies, and measure all three above-mentioned performances metrics $\overline{u}_\Sigma$, $\overline{W}_\text{q}$ and $\overline{P}(A)$ for every strategy in both reference scenarios $1$ and $2$. Similar to the last tests, every strategy is evaluated through a $20$-round Monte-Carlo test where each round begins with a random initial state and lasts $40$ operations periods. Impatient tenants are considered.

To provide benchmarks, we test the controller with two specific ``na\"ive'' strategies: \emph{Prefer Type 1}: the preference vector is $[1, 2, 0]$ at all system states; \emph{Prefer Type 2}: the preference vector is $[2, 1, 0]$ at all system states.
Moreover, we implement and test a simple ``greedy'' single-queue slice admission controller that always accepts the first request in its queue regardless of type, as long as the resource pool supports.

The results are illustrated in Fig.~\ref{fig:effectiveness}. It can be observed that the multi-queuing controller, when specified with an appropriate strategy, outperforms the greedy single-queue solution in admission rate, especially when the demand is dense and queues are congested. However, it shall be noted that the performances highly rely on the selection of strategy, leading to a critical necessity of strategy optimization.

\begin{figure}
	\centering
	\includegraphics[width=.4\textwidth]{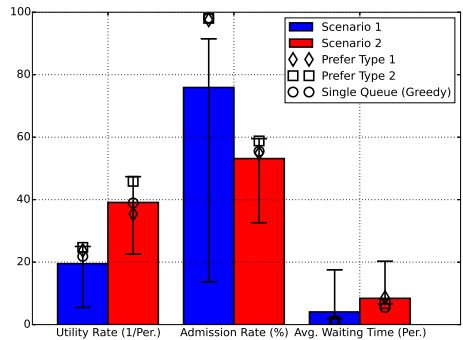}
	\caption{Performance distribution of the proposed multi-queue slice admission controller with $10~000$ random strategies, in comparison to selected benchmarks.}
	\label{fig:effectiveness}
\end{figure}

\change{\section{Further discussion}\label{sect:discussions}}
	%\subsection{Dynamic slicing model}
	In practical wireless networks, both the dynamics of resource availability (e.g. channel fading) and the resource elasticity of active slices must be taken into account. The model in this paper is an approximation with a static resource pool $\mathbf{r}$ and rigid slices, which holds in long-term with appropriate dynamic scheduling to multiplex slices. Note that such a slice multiplexing implicitly enables slice \emph{overbooking} with a risk to break SLAs~\cite{zanzi2018ovens,ovnes2018conext}. The challenge of balancing the multiplexing gain and the overbooking risk in heterogeneous multi-queue admission control settings deserves future study.

%	\subsection{Non-Poisson service scenario}
	It shall also be noticed that the assumptions of Poisson arrivals/releases may not hold in some practical service scenarios. In this case, the queues are not $M/M/1$ systems and cannot be considered as continuous-time Markov systems. Nevertheless, as pointed out in \cite{shortle2018fundamentals},
	%there is imbedded within many continuous-time non-Markov processes a discrete-time Markov chain, when observing the system state only at transition instants. 
	many such continuous-time non-Markov processes can be easily transformed into discrete-time Markov chains by observing only the state transitions.
	Therefore, the analyses given above also apply to most scenarios with non-Poisson request arrivals/releases.

\vspace{+2mm}
\section{Related work}
\label{sect:rel_work}
We summarize in the following the main research efforts in the literature on the topic of Slice-as-a-Service, queuing theory for cloud services and network slicing admission control. 

An overview on multi-tenancy service and 5G network slicing is given in~\cite{samdanis2016network} from perspectives of architecture and standardization, introducing the novel concept of \emph{network slice broker} which executes the admission control. Different attempts have been made in~\cite{sciancalepore2017slice,bega2017optimising} and~\cite{sciancalepore2017mobile} to demonstrate how admission control can benefit the network resource utilization. 

While we have considered network slicing in a generic and abstracted view, which is generally applicable in both radio access network (RAN) and core network (CN) domains, recently there has been a dense specific research interest for RAN slicing and its impact on radio resource management (RRM). On that \cite{sallent2017radio} and \cite{vo2018slicing} provide interesting solutions for efficient resource management and orchestration.
From the perspective of slicing admission strategy optimization, the methods reported in~\cite{bega2017optimising,sciancalepore2017slice,han2018slice} can be worthwhile to refer. Although all these works only consider a binary decision mechanism where declined requests simply vanish instead of being served after a delay, the algorithms deployed by them to solve ILP problems will inspire future development of model-less heuristic strategy optimizers for the proposed multi-queue slice admission controller.

SlaaS shall be considered as a specific type of public cloud environment, where service sessions can be categorized into multiple types with significantly heterogeneous resource demands. Queuing theory has been widely applied for cloud computing services to model the statistics of service demand and delivered quality of service (QoS), such as~\cite{vilaplana2014queuing} and \cite{chang2016modeling}. Especially, service schedulers with heterogeneous queues for different service types are discussed in~\cite{li2017qos} and~\cite{guo2018optimal}. These models provide valuable reference views in addition to the model proposed in this paper. 
Finally, balking and reneging behavior of impatient clients in queuing systems are extensively studied in~\cite{bocquet2005queueing,yue2008waiting}.
%~\cite{ancker1963some1,ancker1963some2,bocquet2005queueing} and \cite{yue2008waiting}.

Differing from the aforementioned works wherein a ``strategy'' usually represents the decision as a function of the system state, our study proposes a novel mechanism of multi-queuing slice admission control where the slicing strategy represents the MNO's preference of slice types in different system states. Besides, out paper also considers impatient tenants, which, from the best of our knowledge, has never been investigated in SlaaS environments.

\section{Conclusion}
\label{sect:concl}
The network slicing paradigm plays a key-role in the next generation of networks design. However, it involves a number of challenges while devising an admission control solution that takes into account complex network tenants behaviors. 

In this paper, we have proposed a multi-queue-based controller that automatically accounts for tenants waiting to get their requests network slices with given request frequency and patience characteristics. Our results validate the proposed model showing that unexpected tenants behaviors may be properly addressed with advanced admission control policies.

\section*{Acknowledgments}
%The work of B. Han and H. D. Schotten was supported in part by the European Union Horizon-2020 Project 5G-MoNArch under Grant Agreement 761445, in part by the Network for the Promotion of Young Scientists (TU-Nachwuchsring), Technische Universit\"at Kaiserslautern with individual funding. The work of V. Sciancalepore and X. Costa-Perez was supported by the European Union Horizon-2020 Project 5G-Transformer under Grant Agreement 761536.
This work has been partially funded by the European Union Horizon-2020 Projects 5G-MoNArch and 5G-Transformer under Grant Agreements 761445 and 761536 as well as by the Network for the Promotion of Young Scientists (TU-Nachwuchsring), TU Kaiserslautern with individual funding.

%\appendix[Sequence of Geometrically Distributed Events as Geometrically Distributed Event]\label{app:geom_seq}
%\begin{lemma}\label{lemma:geom_seq}
%	Given a finite set of random events $v_1,v_2,\dots,v_N$, whose number of arrivals are mutually independent geometric random processes $X_1(p_1),X_2(p_2),\dots,X_N(p_N)$, respectively. Then for an arbitrary ordered sequence $\mathbf{e}=[e_1,e_2,\dots,e_K]$ where $K\in\mathbb{N}$ and $e_i\in\{v_1,v_2,\dots,v_N\}, \forall i\in\{1,2,\dots,K\}$, its arrival counting process $X_\mathbf{e}$ is also geometric.
%\end{lemma}
%\begin{proof}
%	
%	\todo{I've done this on paper, need to clean it and put it here}
%\end{proof}
%%\appendices
%%\section{Proof of the First Zonklar Equation}
%%Appendix one text goes here.
%%\section{}
%%Appendix two text goes here.

%% use section* for acknowledgment
%\section*{Acknowledgment}
%\todo{TU-Nachwuchsring, 5G-MoNArch(?), 5GAuRA(?)}

%\ifCLASSOPTIONcaptionsoff
%  \newpage
%\fi


\bibliographystyle{IEEEtran}
\bibliography{references}

\begin{IEEEbiography}{Bin Han}

\end{IEEEbiography}

%\vfill
\begin{IEEEbiography}{Vincenzo Sciancalepore}
	
\end{IEEEbiography}

%\vfill
\begin{IEEEbiography}{Hans D. Schotten}

\end{IEEEbiography}


\end{document}